\documentclass[english,conference]{IEEEtran}
\usepackage[T1]{fontenc}
\usepackage[latin9]{inputenc}
\usepackage{amsmath}
\usepackage{graphicx}
\usepackage{amssymb}

\makeatletter
\newtheorem{definitn}{Definition}
\newtheorem{thm}{Theorem}
\newtheorem{lemma}{Lemma}

\ifCLASSINFOpdf% \usepackage[pdftex]{graphicx}
\else% or other class option (dvipsone, dvipdf, if not using dvips). graphicx
\fi% graphicx was written by David Carlisle and Sebastian Rahtz. It is
\makeatother

\usepackage{babel}

\begin{document}

\title{Note on Noisy Group Testing: Asymptotic Bounds and Belief Propagation
Reconstruction}

\author{\IEEEauthorblockN{Dino Sejdinovic, Oliver Johnson} \IEEEauthorblockA{School
of Mathematics\\
 University of Bristol\\
 University Walk, Bristol BS8 1TW, UK\\
 Email: \{d.sejdinovic, o.johnson\}@bristol.ac.uk} }
\maketitle
\begin{abstract}
An information theoretic perspective on group testing problems has
recently been proposed by Atia and Saligrama, in order to characterise
the optimal number of tests. Their results hold in the noiseless case,
where only false positives occur, and where only false negatives occur.
We extend their results to a model containing both false positives
and false negatives, developing simple information theoretic bounds
on the number of tests required. Based on these bounds, we obtain
an improved order of convergence in the case of false negatives only.
Since these results are based on (computationally infeasible) joint
typicality decoding, we propose a belief propagation algorithm for
the detection of defective items and compare its actual performance
to the theoretical bounds. 
\end{abstract}

\section{Introduction and Problem Outline}

\thispagestyle{empty}The idea of group testing was introduced during
World War II in order to reduce the cost of large scale blood tests
by pooling blood samples together \cite{Dorfman1943}. Since then,
it emerged as a promising approach in various applications, including
multiple access communications and DNA clone library screening (cf.
\cite{Du2000} and references therein).

The advent of compressed sensing (CS) has revived interest in group
testing \cite{Cheraghchi2009}, as both problems involve the detection
of a sparse high-dimensional signal via a small number of random measurements.
However, the compressed sensing literature has mostly focussed on
problems with measurement matrices with entries taken from distributions
with densities. Group testing naturally belongs in a broader framework
of discrete compressed sensing, where the entries are random integers,
often just 0s and 1s. This framework of discrete compressed sensing
includes wider applications such as genotyping \cite{Erlich2010}.
An extension of the group testing problem to the scenario where pools
must conform to the constraints imposed by a graph has also been studied
recently \cite{Cheraghchi2010}.

An information theoretic approach to a noisy version of group testing
was recently developed by Atia and Saligrama \cite{Atia2009a,Atia2009}.
We adopt much of their model and notation, which we will first briefly
review. We will use group testing to identify $K$ defective items
within a larger collection of $N$ items, by testing a pool of items
at a time. Each test reveals whether the pool contains any defective
items, i.e., the test result is \textit{positive}, or 1, if at least
one of the items in the pool is defective, and it is otherwise \textit{negative},
or 0. However, we will allow two types of errors to occur in the testing. 
\begin{enumerate}
\item \textbf{False positives}, where the test result is positive with probability
$q$ when the pool does not contain any defective items. In other
words, the result of the test is ORed with Bernoulli($q$) random
variable. 
\item \textbf{False negatives} - the indicator whether an item is defective
is {}``diluted'' with probability $u$. In other words, the result
of the test will only be positive if the indicator of some defective
item passes through a $Z$-channel successfully. 
\end{enumerate}
Note that the false negatives make the analysis significantly more
complicated than for standard coding theoretic problems. This is because
this model makes the noise dependent on the input, because a pool
with more defective items will be less likely to return a false negative.
Interestingly, our results here indicate that false negatives are,
in a certain sense, easier to deal with than false positives, and
we obtain an improved order of convergence on the number of achievable
tests in the case of false negatives only. 

Another way to describe the false negative process is that the test
will be positive if the sum of the indicators, thinned in the sense
of Rényi, is positive. In future work we hope to explore whether the
bounds on entropy under thinning proved in \cite{Johnson2009} can
improve or generalize the results of this paper.

We now formally describe the model which incorporates the presence
of these two kinds of testing errors. 
\begin{definitn}
\label{def:model} Let $\beta\in\{0,1\}^{N}$ be a column vector of
indicators corresponding to the overall set of items, i.e., $\beta_{i}=1$
iff item $i$ is defective. We consider the case $w(\beta)=K\ll N$,
where $w(\cdot)$ represents the Hamming weight. Furthermore, $\mathbf{X}=(x_{ti})\in\{0,1\}^{T\times N}$
will denote the measurement matrix, s.t. $x_{ti}=1$ iff item $i$
is pooled in test $t$. We will restrict our attention to the case
where $\mathbf{X}$ is composed of i.i.d. Bern($p$) entries. A set
of test results is a vector $y\in\{0,1\}^{T}$, where $y_{t}=1$ means
test $t$ is positive. The outcome $y_{t}$ of test $t$ is given
by (symbol $\wedge$ stands for Boolean matrix product): \begin{equation}
y_{t}=(x_{t}\wedge\mathbf{D}_{t}\wedge\beta)\vee z_{t},\; t\in\{1,2,\ldots,T\}.\label{eq: test_outcome}\end{equation}
 Here $x_{t}$ denotes the $t$-th row of $\mathbf{X}$, $\mathbf{D}_{t}\in\{0,1\}^{N\times N}$
is a diagonal matrix with i.i.d. Bern($1-u$) entries on the diagonal,
independent of $\beta$ and $\mathbf{X}$, and $z_{t}$ is a $Bern(q)$
random variable, independent of all others. 
\end{definitn}
This compact notation captures both the false positive test results
which occur when $z_{t}=1$, and the false negative test results which
occur in the event that all the diagonal entries of $\mathbf{D}_{t}$
corresponding to the defective items in pool $t$ equal zero.

In \cite{Atia2009a}, Atia and Saligrama showed how group testing
can be viewed analogously to channel coding by considering a set of
$K$ channels, with input $X_{(i)}$ and the pair $(X_{(K-i)},Y)$
as their output, $i\in\{1,2,\ldots,K\}$. Here, $X_{(i)}$ stands
for the $i$ entries in the row of the measurement matrix corresponding
to (any) $i$ defective items, and $Y$ is the test outcome (viewed
as a random variable). Atia and Saligrama prove the following result: 
\begin{thm}
(\cite{Atia2009a}, Theorem 3.2) \label{thm:atiamain} Consider the
joint typicality decoder in the model of Definition \ref{def:model},
in the case where one or both of $q=0$ or $u=0$ (the model is noisefree,
or allows false positives or false negatives, but not both). An achievable
number of tests $T_{typ}$ which allows perfect detection is given
by: \begin{equation}
T_{typ}=\max_{i}\frac{\log_{2}\binom{N-K}{i}\binom{K}{i}}{I(X_{(i)};X_{(K-i)},Y)}.\label{eq:ratio}\end{equation}

\end{thm}
We now describe the structure of the remainder of this paper. We will
consider the model of Definition \ref{def:model}, in the case where
both $q$ and $u$ can be non-zero (both false positives and false
negatives are allowed). For reasons of space, we will assume that
an analogue of Equation \eqref{eq:ratio} holds in the case $q>0$
and $u>0$. (To verify this requires a somewhat lengthy analysis of
the probability that $X$ and $Y$ are jointly typical, as performed
in the Appendix of \cite{Atia2009a}). This means that the key quantity
of interest is the mutual information $I(X_{(i)};X_{(K-i)},Y)$. We
will analyse this quantity in Section \ref{sec:asymp}, as in \cite{Atia2009a}
deducing asymptotic results of the form $T_{typ}=\mathcal{O}(K\log(K(N-K))$.
We also deduce that in the case where only false negatives occur,
the number of tests required can be reduced by a factor of $\log K$.

In Section \ref{sec:bp}, we propose a belief propagation algorithm
for the detection of the defective items in noisy group testing. The
analysis of Theorem \ref{thm:atiamain} is based on the use of a joint
typicality decoder, which is infeasible in practice, having prohibitive
computational complexity in the limit of large $K$ and $N$. Belief
propagation offers a practically implementable alternative. Belief
propagation has previously been used in the statistical physics community
to address the problem of noiseless group testing and its relationship
to the hitting set problem \cite{Mezard2007}.

\section{Asymptotic Bounds}

\label{sec:asymp}

In this section, we will derive sharp bounds on the mutual information
of Equation (\ref{eq:ratio}): 
\begin{lemma}
\label{lem:identity} The mutual information $I(X_{(i)};X_{(K-i)},Y)$
can be expressed in closed form as $I_{1}+I_{2}$, where the {}``lead
term'' is:

\begin{multline}
I_{1}=i(1-q)(1-p+pu)^{K}\cdot\\
\left(\frac{pu}{1-p+pu}\log_{2}u-\log_{2}(1-p+pu)\right).\label{eq:i1value}\end{multline}

and the {}``error term'' is:

\begin{multline}
I_{2}=\frac{1}{\log2}\sum_{j=2}^{\infty}\Biggl[\frac{(1-q)^{j}}{j(j-1)}.\\
(1-p+pu^{j})^{K}\Biggl(1-\biggl(\frac{(1-p+pu)^{j}}{1-p+pu^{j}}\biggr)^{i}\Biggr)\Biggr]\label{eq:i2value}\end{multline}
\end{lemma}
\begin{proof}
As in \cite{Atia2009a}, we decompose \begin{equation}
I(X_{(i)};X_{(K-i)},Y)=H(Y|X_{(K-i)})-H(Y|X_{(K)}),\label{eq:mutinf}\end{equation}
 and consider the two terms separately. First, we set $V=X\wedge\mathbf{D}\wedge\beta$,
and notice that

\begin{multline}
\mathbb{P}(Y=0|w(X_{(K)})=j)=\\
\mathbb{P}(Z=0)\mathbb{P}(V=0|w(X_{(K)})=j)=\\
(1-q)u^{j},\label{eq:prob1}\end{multline}
which means that, writing $h(\cdot)$ for the binary entropy function,
\begin{equation}
H(Y|X_{(K)})=\sum_{j=0}^{K}\binom{K}{j}p^{j}(1-p)^{K-j}h\left[(1-q)u^{j}\right].\label{eq:firstterm}\end{equation}
 Similarly,

\begin{multline}
\mathbb{P}(Y=0|w(X_{(K-i)})=l)=\\
\mathbb{P}(Z=0)\mathbb{P}(V=0|w(X_{(K-i)})=l)=\\
(1-q)\sum_{j=l}^{l+i}\mathbb{\Bigl[P}(V=0|w(X_{(K)})=j)\cdot\\
\mathbb{P}(w(X_{(K)})=j|w(X_{(K-i)})=l)\Bigr]=\\
(1-q)\sum_{j=l}^{l+i}u^{j}\mathbb{P}(w(X_{(i)})=j-l)=\\
(1-q)u^{l}\sum_{j=0}^{i}u^{j}\binom{i}{j}p^{j}(1-p)^{i-j}=\\
(1-q)u^{l}(1-p+pu)^{i},\label{eq: derivation}\end{multline}
 whereby we obtain:

\begin{multline}
H(Y|X_{(K-i)})=\sum_{l=0}^{K-i}\Biggl[\binom{K-i}{l}p^{l}(1-p)^{K-i-l}\cdot\\
h\left((1-q)u^{l}(1-p+pu)^{i}\right)\Biggr].\label{eq:secondterm-1}\end{multline}

We substitute the expressions (\ref{eq:firstterm}) and (\ref{eq:secondterm-1})
into Equation (\ref{eq:mutinf}), and analyse the resulting sum. We
use an expansion of binary entropy as

\begin{equation}
h(\theta)=\left(-\theta\log_{2}\theta\right)+\frac{1}{\log2}\left(\theta-\sum_{j=2}^{\infty}\frac{\theta^{j}}{j(j-1)}\right),\label{eq: bin_entropy}\end{equation}
with the first bracketed term becoming (\ref{eq:i1value}), and the
remaining expression becoming (\ref{eq:i2valueq0}). 
\end{proof}
Observe that for any $j\geq2$, the function $g(p)=(1-p+pu^{j})-(1-p+pu)^{j}\geq0$.
This means that the bracketed term in (\ref{eq:i2valueq0}) is positive,
and so as in \cite{Atia2009a}, we could simply use the lower bound
$I(X_{(i)};X_{(K-i)},Y)\geq I_{1}$ in Theorem \ref{thm:atiamain}.
However, in many cases $I_{2}$ turns out to play a significant role,
and so by including it in our analysis we obtain better bounds. 
\begin{lemma}
\label{lem:leading} Choosing $p=(1-u)^{-1}/K$, we obtain that for
constant $i$, $q$ and $u$

\begin{equation}
I_{1}=\frac{i(1-u)(1-q)(u\log u-u+1)}{Ke\log2}+\mathcal{O}\left(\frac{1}{K^{2}}\right).\end{equation}
\end{lemma}
\begin{proof}
By setting $p=\alpha/K$, and expanding Equation (\ref{eq:i1value})
in powers of $1/K$, we obtain

\begin{equation}
I_{1}=\frac{\alpha i(1-q)e^{\alpha(u-1)}(u\log u-u+1)}{Ke\log2}+\mathcal{O}\left(\frac{1}{K^{2}}\right).\end{equation}
We can optimize this expression over $\alpha$ by taking $\alpha=1/(1-u)$,
which justifies the heuristic choice of $p=1/K$ to define the measurement
matrices in \cite{Atia2009a}. 
\end{proof}
Note that in the cases $u=q=0$ and $u=0$ respectively we recover
$i/(Ke\log2)$ from (15) of \cite{Atia2009a} and $i(1-q)/(Ke\log2)$
from (29) of \cite{Atia2009a}. In the case $q=0$, this optimal choice
of $\alpha$ gives us a lower bound of $i(1-u)/(2Ke\log2)$, a slight
improvement on (37) of \cite{Atia2009a}.

Similarly, it can be shown that with $p=\alpha/K$,

\begin{multline}
\lim_{K\to\infty}KI_{2}=\frac{\alpha i}{\log2}\sum_{j=2}^{\infty}\Biggl[\frac{(1-q)^{j}}{j(j-1)}\cdot\\
e^{\alpha u^{j}-\alpha}(u^{j}+j-ju-1)\Biggr].\label{eq: series}\end{multline}
It is easy to see that the series in \eqref{eq: series} is converging
for $q\neq0$. Furthermore, by repeatedly using the sum \begin{eqnarray}
\sum_{j=2}^{\infty}\frac{\theta^{j}}{j(j-1)} & = & \theta+(1-\theta)\log(1-\theta),\label{eq: sum_theta}\end{eqnarray}
 and the obvious inequality $0\leq u^{j}\leq u^{2}$, for $j\geq2$,
we obtain:

\begin{equation}
\frac{\alpha e^{-\alpha}i}{\log2}C_{q,u}\leq\lim_{K\to\infty}KI_{2}\leq\frac{\alpha e^{-\alpha}i}{\log2}e^{\alpha u^{2}}C_{q,u},\end{equation}
where

\begin{equation}
C_{q,u}=q-(1-u+qu)(1+\log q-\log(1-u+qu)).\end{equation}
Notice that $C_{q,u}=\infty$ when $q=0$. This suggests that $I_{2}$
is of a larger order in this case. Indeed, the following Lemma holds:
\begin{lemma}
\label{lem:qzero}In case $q=0$, $I_{2}=\mathcal{O}(\frac{\log K}{K})$.
In particular,

\begin{equation}
\frac{\alpha e^{-\alpha}i}{\log2}(1-u)\leq\lim_{K\to\infty}\frac{K}{\log K}I_{2}\leq\frac{\alpha e^{-\alpha}i}{\log2}e^{\alpha u^{2}}(1-u).\end{equation}
\end{lemma}
\begin{proof}
In

\begin{multline}
I_{2}=\frac{1}{\log2}\sum_{j=2}^{\infty}\Biggl[\frac{(1-p+pu^{j})^{K}}{j(j-1)}.\\
\Biggl(1-\biggl(\frac{(1-p+pu)^{j}}{1-p+pu^{j}}\biggr)^{i}\Biggr)\Biggr]\label{eq:i2valueq0}\end{multline}
we notice that $(1-p+pu^{j})^{K}\uparrow e^{-\alpha(1-u^{j})}\leq e^{-\alpha(1-u^{2})},$
as $K\to\infty$, whereas $(1-p+pu^{j})^{K}\geq(1-p)^{K}$. Therefore,
by applying \eqref{eq: sum_theta} with $\theta=(1-p+pu)^{i}$, we
obtain that $\forall i,K$,

\begin{multline}
\frac{(1-p)^{K}}{\log2}\Biggl[1-\frac{(1-p+pu)^{i}}{(1-p)^{i}}+\\
\frac{(1-(1-p+pu)^{i})\log(1-(1-p+pu)^{i})}{(1-p)^{i}}\Biggr]\leq\\
\leq I_{2}\leq\\
\leq\frac{e^{-\alpha(1-u^{2})}}{\log2}\Biggl[1-\frac{(1-p+pu)^{i}}{(1-p+pu^{2})^{i}}+\\
\frac{(1-(1-p+pu)^{i})\log(1-(1-p+pu)^{i})}{(1-p+pu^{2})^{i}}\Biggr].\label{eq: I2bound}\end{multline}
Now, by developing both sides in powers of $1/K$, 

\begin{multline}
\frac{\alpha e^{-\alpha}i}{\log2}\cdot\\
\left(\frac{(1-u)\left[\log K-\log(\alpha i-\alpha iu)\right]-u}{K}\right)+\mathcal{O}\left(\frac{1}{K^{2}}\right)\leq\\
\leq I_{2}\leq\\
\leq\frac{\alpha e^{-\alpha(1-u^{2})}i}{\log2}\cdot\\
\left(\frac{(1-u)\left[\log K-\log(\alpha i-\alpha iu)\right]-u+u^{2}}{K}\right)+\\
\mathcal{O}\left(\frac{1}{K^{2}}\right),\label{eq: I2bound2}\end{multline}
which proves the claim.\end{proof}
\begin{thm}
Assuming that an equivalent of Theorem \ref{thm:atiamain} holds in
the general case, using Lemmas \ref{lem:identity}, \ref{lem:leading}
and \ref{lem:qzero}, we deduce: 

(i) For any $q>0$, $u\geq0$, as $K\rightarrow\infty$,

\begin{equation}
T_{typ}=\mathcal{O}(K\log(K(N-K)).\end{equation}
 In particular,

\begin{multline}
\alpha e^{-\alpha}C_{q,u}\leq\\
\lim_{K\to\infty}\biggl(\frac{K\log(K(N-K))}{T_{typ}}-\\
\alpha\left[e^{-\alpha(1-u)}(1-q)(u\log u-u+1)\right]\biggr)\leq\\
\alpha e^{-\alpha(1-u^{2})}C_{q,u}.\label{eq: 01}\end{multline}

(ii) For $q=0$, and any $u\geq0$, as $K\to\infty$,

\begin{equation}
T_{typ}=\mathcal{O}\left(K(1+\frac{\log(N-K)}{\log K})\right).\label{eq: 02}\end{equation}

In particular,

\begin{multline}
\alpha e^{-\alpha}(1-u)\leq\\
\lim_{K\to\infty}\frac{K\log(K(N-K))}{T_{typ}\log K}\leq\\
\alpha e^{-\alpha(1-u^{2})}(1-u).\label{eq: 03}\end{multline}

\end{thm}
Notice that in the case of noiseless group testing, i.e., $u=0$,
$q=0$, we arrive at the exact asymptotic expressions for $T_{typ}$:

\begin{equation}
T_{typ}=eK(1+\frac{\log(N-K)}{\log K}).\end{equation}

In the noisy case, the derived bounds are sharp. Fig. \ref{fig:bounds}
depicts the quantities bounding $\lim_{K\to\infty}\frac{K\log(K(N-K))}{T_{typ}}$
as a function of $\alpha$ for $u=0.05$ and $q=0.01$. The bounds
coincide in the first two decimal places.

\begin{figure}

\includegraphics[width=3.6in]{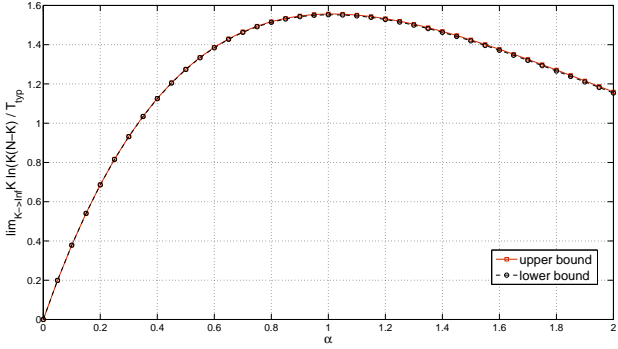}

\caption{\label{fig:bounds}The bounds on the constant in the asymptotic estimate
$T_{typ}=\mathcal{O}(K\log(K(N-K))$ for $u=0.05$, $q=0.01$.}

\end{figure}

\section{Belief Propagation Reconstruction}

\label{sec:bp} The joint typicality decoder analysed in Section \ref{sec:asymp}
has prohibitive computational complexity in the limit of large $K$
and $N$. In this section, we compare the theoretical performance
with the performance of belief propagation (BP) decoder, which performs
an approximate bitwise MAP (maximum a posteriori) detection of defective
items by solving:

\begin{equation}
\hat{\beta}_{i}^{(MAP)}=\arg\max_{\beta_{i}\in\{0,1\}}\mathbb{\mathbb{P}}(\beta_{i}|y),\; i\in\{1,2,\ldots,N\}.\end{equation}

The above can be transformed into:

\begin{eqnarray}
\hat{\beta}_{i}^{(MAP)} & = & \arg\max_{\beta_{i}\in\{0,1\}}\sum_{\sim\beta_{i}}\left[\prod_{t=1}^{T}\mathbb{P}(y_{t}|\beta_{supp(y_{t})})\prod_{j=1}^{N}\mathbb{P}(\beta_{j})\right]\nonumber \\
 & = & \arg\max_{\beta_{i}\in\{0,1\}}\sum_{\sim\beta_{i}}\Biggl[\prod_{t=1}^{T}\mathbb{P}(y_{t}|w(\beta_{supp(y_{t})}))\cdot\nonumber \\
 &  & \prod_{j=1}^{N}\left(\lambda\delta_{\beta_{j}}(1)+(1-\lambda)\delta_{\beta_{j}}(0)\right)\Biggr],\label{eq: MAP01}\end{eqnarray}
where $\lambda=K/N$. Therefore, MAP detection amounts to the marginalisation
of a function which permits a sparse factorisation, and as such can
be performed efficiently via message passing on a factor graph corresponding
to the measurement matrix $\mathbf{X}$.

The belief propagation message-update rules are given by:

\begin{multline}
m_{i\to t}^{(l+1)}(\beta_{i})\propto\left(\lambda\delta_{\beta_{i}}(1)+(1-\lambda)\delta_{\beta_{i}}(0)\right)\cdot\\
\prod_{b\in\mathcal{N}(i)\backslash\{t\}}\hat{m}_{b\to i}^{(l)}(\beta_{i}),\label{eq: BPupdate1}\end{multline}

\begin{multline}
\hat{m}_{t\to i}^{(l)}(\beta_{i})\propto\sum_{\sim\beta_{i}}\Biggl[\mathbb{P}(y_{t}|w(\beta_{supp(y_{t})}))\cdot\\
\prod_{j\in\mathcal{N}(t)\backslash\{i\}}m_{j\to t}^{(l)}(\beta_{j})\Biggr].\label{eq: BPupdate2}\end{multline}

The fact that $\mathbb{P}(y_{t}|\beta_{supp(y_{t})})=\mathbb{P}(y_{t}|w(\beta_{supp(y_{t})}))$
greatly simplifies the message-passing update rules. In particular,
since $\mathbb{P}(y_{t}|w(\beta_{supp(y_{t})}))=(1-q)u^{w(\beta_{supp(y_{t})})}$
due to the symmetry between $x_{t}$ and $\beta$ in \eqref{eq: test_outcome},
the above equations, by rewriting message-update rules in terms of
log-ratios, i.e.,

\begin{eqnarray}
L_{i\to t}^{(l)}=\log\frac{m_{i\to t}^{(l)}(1)}{m_{i\to t}^{(l)}(0)} & ,\; & \hat{L}_{t\to i}^{(l)}=\log\frac{\hat{m}_{t\to i}^{(l)}(1)}{\hat{m}_{t\to i}^{(l)}(0)}.\label{eq: log-ratios}\end{eqnarray}
simplify to:

\begin{eqnarray}
L_{i\to t}^{(l)} & = & \begin{cases}
\log\frac{\lambda}{1-\lambda}, & l=0,\\
\log\frac{\lambda}{1-\lambda}+\sum_{b\in\mathcal{N}(i)\backslash\{t\}}\hat{L}_{b\to i}^{(l)}, & l\geq1,\end{cases}\label{eq: updateItoT}\end{eqnarray}
and

\begin{multline}
\hat{L}_{t\to i}^{(l)}=\\
\log\left(u+\frac{1-u}{1-(1-q)\prod_{j\in\mathcal{N}(t)\backslash\{i\}}\left(u+\frac{1-u}{1+\exp(L_{j\to t}^{(l)})}\right)}\right),\label{eq: updateTtoI1}\end{multline}
in the case of a positive $t$-th test, i.e., when $y_{t}=1$, and
simply

\begin{eqnarray}
\hat{L}_{t\to i}^{(l)} & = & \log u,\label{eq: updateTtoI2}\end{eqnarray}
for $y_{t}=0$.

In a preliminary assessment of the belief propagation reconstruction,
we simulated BP decoder for noisy group testing in the case where
$N=5000$, $K=50$, $u=0.05$, $q=0.01$, and for various values of
parameter $p$. We performed at least 200 trials at the various numbers
of tests. The number of iterations was fixed to $50$. 

As illustrated in Fig. \ref{fig: simA}, the detected probability
of perfect reconstruction increases with $p$, and is about 99\% when
the number of tests was $T\approx1600$ for $p=0.02$. The value of
$p$ which performs best here is $1/K$, suggesting that the same
heuristics concerning the optimal $p$ discussed for typical set decoding
also apply for belief propagation. 

In Fig. \ref{fig: simB}, we illustrated the number of detection errors
per size of the support as a function of the number of tests. This
figure illustrates that even though a large probability of perfect
reconstruction is achieved only at the relatively large number of
tests, the BP decoder typically diagnoses only a few items incorrectly
at the number of tests as small as $T\approx900$. These results are
still far from the estimate arising from the asymptotic analysis of
joint typicality decoder in the previous section, which is $T_{typ}\approx400$,
but nonetheless confirm the utility of the belief propagation algorithm
in noisy group testing, even though no design of the measurement matrix
that complies well with belief propagation algorithm has been taken
into consideration. It may also be possible to achieve further improvements
in performance by using belief propagation with decimation as in \cite{Mezard2007}.

\begin{figure}

\includegraphics[width=3.5in]{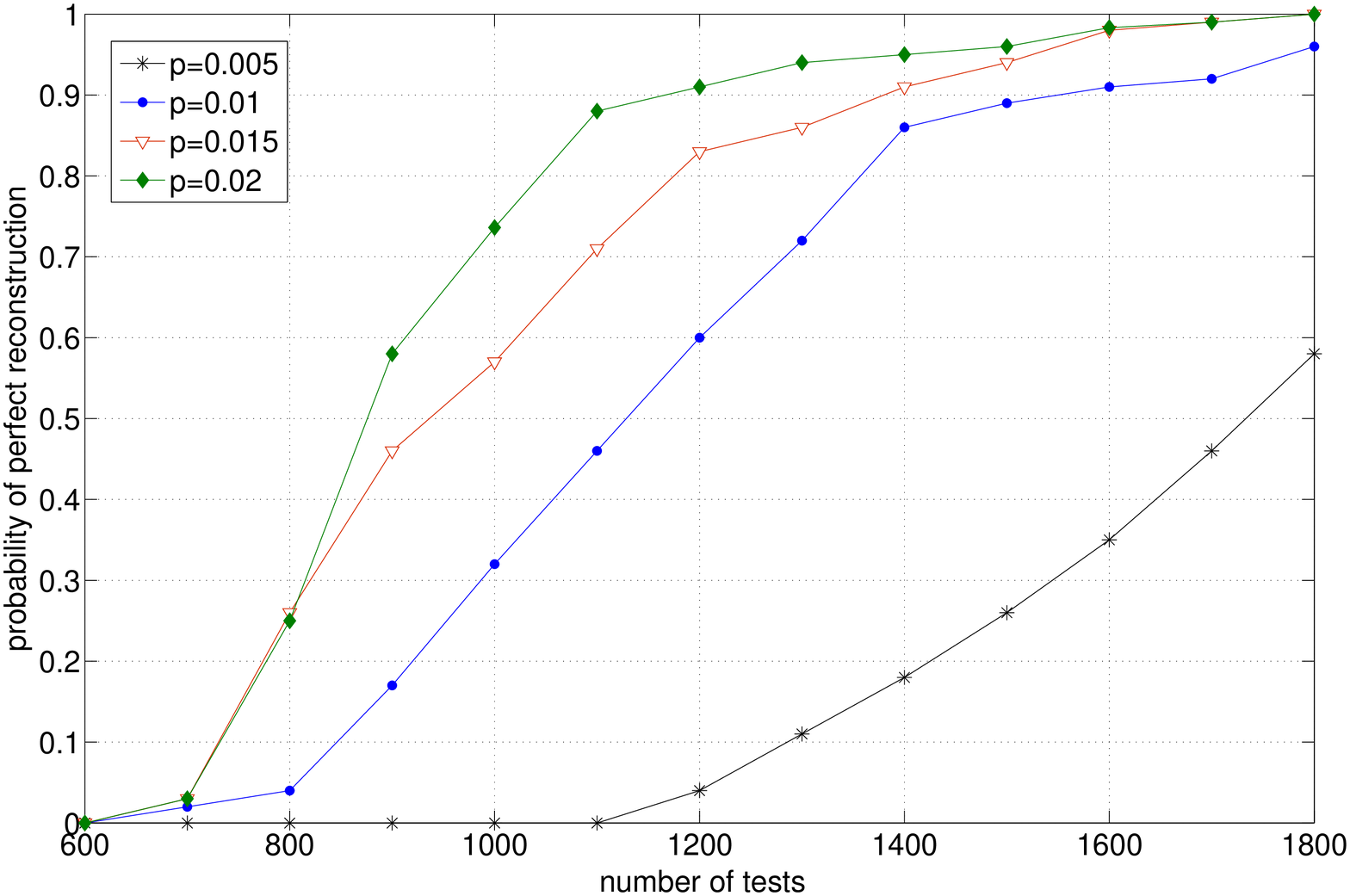}

\caption{\label{fig: simA}Probability of perfect reconstruction with BP at
$N=5000$, $K=50$.}

\end{figure}

\begin{figure}

\includegraphics[width=3.7in]{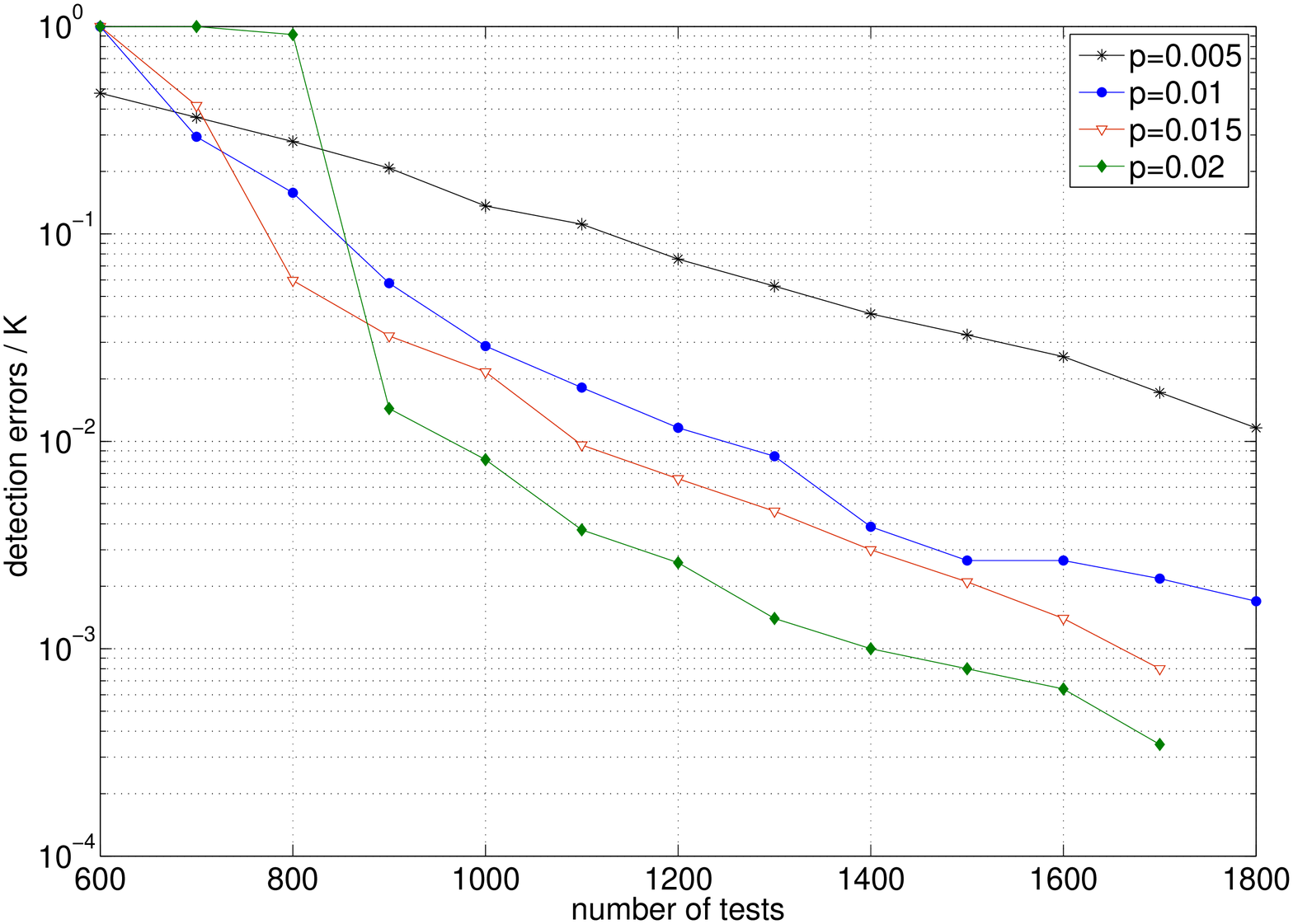}

\caption{\label{fig: simB}The number of detection errors per size of the support
with BP at $N=5000$, $K=50$.}

\end{figure}

\section{Conclusions}

This contribution studies the information theoretic bounds arising
in the problem of noisy group testing and proposes an efficient algorithm
for noisy group testing based on belief propagation. We develop a
sharp estimate on the constants arising in the asymptotic approximation
of the number of tests sufficient for the perfect detection via a
joint typicality decoder, as a function of the noise parameters. We
show how the presence of the false positives in the noisy group testing
changes the order of the achievable number of tests. These result
allows us to benchmark the performance of a belief propagation algorithm.
We restrict our attention here to the case where the measurement matrix
is composed of i.i.d. Bernoulli entries. More general measurement
matrices can be studied in a similar manner, in particular those with
row weights generated according to a pre-optimised degree distribution.
A judicious choice of degree distributions may further improve the
performance of the belief propagation algorithm, in analogy with well
known results in sparse graph coding.

\bibliographystyle{plain}
\bibliography{FullBibliography}

\end{document}